\documentclass[11pt]{article}
\usepackage{geometry}
\usepackage{amsmath,amssymb,amsthm}
\usepackage{graphicx}
\usepackage{bm}
\usepackage[linesnumbered,ruled]{algorithm2e}
\SetKwRepeat{Do}{do}{while}
\usepackage{fullpage}
\usepackage{mathtools}
\usepackage{url}
\newtheorem{theorem}{Theorem}
\newtheorem{lemma}{Lemma}

\newtheorem{claim}{Claim}
\newtheorem{definition}{Definition}

\begin{document}

\title{A Linear Programming Approach to the Super-Stable Roommates Problem\thanks{%
This work was supported by JST ERATO Grant Number JPMJER2301, Japan.}}
\author{Naoyuki Kamiyama}
\date{Institute of Mathematics for Industry \\ Kyushu University, Fukuoka, Japan \\
\url{kamiyama@imi.kyushu-u.ac.jp}}

\maketitle

\begin{abstract}
The stable roommates problem is a non-bipartite version 
of the well-known 
stable matching problem. 
Teo and Sethuraman proved that, 
for each instance of the stable roommates problem
in a complete graph, 
there exists a linear inequality system
such that there exists a feasible solution to 
this system if and only if 
there exists a stable matching in the given instance.
The aim of this paper is to extend the result of 
Teo and Sethuraman 
to the stable roommates problem with ties. 
More concretely, we prove that, 
for each instance of the stable roommates problem with ties
in a complete graph,  
there exists a linear inequality system
such that there exists a feasible solution to 
this system if and only if 
there exists a super-stable matching in the given instance.
\end{abstract} 

\section{Introduction}

The stable roommates problem is a non-bipartite version 
of the stable matching problem in a 
bipartite 
graph~\cite{GaleS62}. 
In contrast to the stable matching problem, 
it is known that
there exists an instance of 
the stable roommates 
problem
where a stable matching does not exist~\cite[Example~3]{GaleS62}. 
Thus, the problem of checking the existence of a stable matching in
a given instance of the stable roommates problem is one of the most 
important problems in the study of the stable roommates problem.
For this problem, 
Irving~\cite{Irving85} proposed a polynomial-time algorithm.
This algorithm is combinatorial, i.e., this 
does not need solving linear programs. 
On the other hand, 
in the study of the stable matching problem, 
linear programming approaches have been actively 
studied (see, e.g., \cite{Fleiner03,Rothblum92,VandeVate89}). 
Thus, it is natural to investigate a linear programming 
approach to the stable roommates problem. 
In this direction, 
Teo and Sethuraman~\cite{TeoS98}
proposed the following linear programming 
approach.
They proved that, 
for each instance of the stable roommates problem
in a complete graph, 
there exists a linear inequality system
such that there exists a feasible solution to 
this system if and only if 
there exists a stable matching in the given instance.
Since a linear program such that 
the separation problem can be solved in 
polynomial time can be solved in polynomial time~\cite{GrotschelLS93}, 
the result of Teo and Sethuraman~\cite{TeoS98} 
can be regarded as another proof of the polynomial-time 
solvability of 
the problem of checking the existence of a stable matching in
a given instance of the stable roommates problem
in a complete graph. 
(See \cite{SethuramanT01,TeoS00} for related work of 
\cite{TeoS98} and \cite{AbeledoB96,AbeledoR94} for another linear 
programming approach to the stable roommates problem.) 

In this paper, 
we extend the linear programming approach of  
Teo and Sethuraman~\cite{TeoS98} 
to the stable roommates problem with ties. 
In particular, we focus on super-stability in the stable 
roommates problem with ties (see Section~\ref{section:preliminaries} for its formal definition). 
Irving and Manlove~\cite{IrvingM02}
proposed a combinatorial 
polynomial-time 
algorithm for the problem of checking the existence of a super-stable 
matching in a given instance of the stable roommates problem with ties
(see also \cite{FleinerIM07,FleinerIM11}).
In this paper, we prove that, 
for each instance of the stable roommates problem with ties
in a complete graph,  
there exists a linear inequality system
such that there exists a feasible solution to 
this system if and only if 
there exists a super-stable matching in the given instance.
Furthermore, we prove that the separation problem for the above linear inequality 
system 
can be solved in polynomial time, and 
a super-stable matching can be constructed from 
a feasible solution to the system 
in polynomial time.
Thus, the result of this paper 
gives another proof of the 
polynomial-time solvability of 
the problem of checking the existence of a super-stable 
matching in a given instance of the stable roommates problem with ties
in a complete graph.

Our proof basically follows the proof of 
\cite{TeoS98}.
The most remarkable difference between 
our proof and that in \cite{TeoS98}
is the following point. 
In the proof of 
\cite{TeoS98}, we first find a feasible solution to 
the linear inequality 
system.
After that, for each vertex, 
we arrange the incident edges having positive values in decreasing order 
according to its preference.
In the setting of \cite{TeoS98}, 
since the preferences do not contain ties, 
this ordering for each vertex is 
uniquely determined. 
However, in our setting, since 
the preferences contain ties, 
this idea cannot be straightforwardly 
applied. 
In this paper, we resolve this issue 
by proving  
the result corresponding to the self-duality result 
in \cite{AbeledoR94,TeoS98} for a linear programming formulation of 
the stable roommates problem without ties. 
The proof of our result is 
basically the same as 
the self-duality result in \cite{HuG21} for a linear programming 
formulation of the super-stable matching problem in bipartite graphs.
By using this self-duality result, 
we prove that, even if 
the preferences contain ties, 
for each vertex, 
the incident edges having positive values 
can be uniquely arranged  
in decreasing order 
according to its preference (see Lemma~\ref{lemma:strict_super}).

\section{Preliminaries} 
\label{section:preliminaries} 

Let $\mathbb{R}_+$, $\mathbb{Q}_+$ denote the sets of non-negative real numbers
and non-negative rational numbers, respectively. 
For each positive integer $z$, we define $[z] := \{1,2,\dots,z\}$. 
For each finite set $U$, each vector $x \in \mathbb{R}_+^U$, 
and each subset $W \subseteq U$, we define 
$x(W) := \sum_{u \in W}x(u)$. 

In this paper, we are given the complete graph $G = (V,E)$ on 
a finite vertex set $V$ such that $|V|$ is an even positive integer.
Notice that $E = \{e \subseteq V \mid |e| = 2\}$.
For each subset $F \subseteq E$ and 
each vertex $v \in V$, 
we define $F(v) \coloneqq \{e \in F \mid v \in e\}$. 
For each vertex $v \in V$, we are given a transitive binary 
relation $\succsim_v$ 
on $E(v)$ 
such that, 
for every pair of edges $e,f \in E(v)$, at least one of 
$e \succsim_v f$, $f \succsim_v e$ holds. 
For each vertex $v \in V$ and 
each pair of edges $e,f \in E(v)$,
if $e \succsim_v f$ and $f \not\succsim_v e$
(resp.\ $e \succsim_v f$ and $f \succsim_v e$), then
we write $e \succ_v f$
(resp.\ $e \sim_v f$). 
Intuitively speaking, 
if $e \succ_v f$, then 
$v$ prefers $e$ to $f$. 
If $e \sim_v f$, then $v$ is indifferent between $e$ and $f$.

\begin{definition}
A subset $\mu \subseteq E$ is called a \emph{matching in $G$} if
$|\mu(v)| = 1$ for every vertex $v \in V$.
\end{definition} 

For each matching $\mu$ in $G$ and each vertex $v \in V$, 
we do not distinguish between $\mu(v)$ and the edge in $\mu(v)$. 

\begin{definition}
Let $\mu$ be a matching in $G$. 
Then for each edge $e \in E \setminus \mu$, 
we say that 
$e$ \emph{weakly blocks $\mu$}
if 
$e \succsim_v \mu(v)$
for every vertex $v \in e$.
\end{definition} 

\begin{definition}
Let $\mu$ be a matching in $G$.
Then 
$\mu$ is said to be 
\emph{super-stable}
if no edge in $E \setminus \mu$ weakly blocks $\mu$. 
\end{definition}

Let $v$ be a vertex in $V$, and
let $e$ be an edge in $E(v)$. 
For 
each symbol $\odot \in \{\succ_v,\succsim_v,\sim_v\}$, 
we define 
$E[\mathop{\odot} e]$ (resp.\ $E[e \mathop{\odot}]$) as the set of 
edges $f \in E(v)$
such that $f \mathop{\odot} e$
(resp.\ $e \mathop{\odot} f$). 

For a positive integer $k$, 
a sequence $(v_0,v_1,\dots,v_k)$ 
of vertices in $V$ is called a \emph{walk in $G$} if 
$v_{i-1} \neq v_i$ holds 
for every integer $i \in [k]$. 
It should be noted that 
a walk in $G$ can pass through the same vertex more than once. 
For each 
walk $C = (v_0,v_1,\dots,v_k)$ in $G$,
if $v_0 = v_k$, then $C$ is called 
a \emph{closed walk in $G$}. 
For every closed walk $C = (v_0,v_1,\dots,v_k)$ in $G$,
we have $k \ge 2$. 

Let $C = (v_0,v_1,\dots,v_k)$ be a closed 
walk in $G$.
For each edge $e \in E$, 
the \emph{multiplicity ${\sf mult}_C(e)$ of $e$ with respect to $C$} is defined as 
the number of integers $i \in [k]$ such that 
$\{v_{i-1},v_i\} = e$. 
The \emph{multiplicity ${\sf mult}(C)$ of $C$} is defined as  
$\max_{e \in E}{\sf mult}_C(e)$. 
If $k \ge 3$ and 
$v_i \neq v_j$ holds 
for every pair of distinct integers $i,j \in [k]$, then 
$C$ is called a \emph{cycle in $G$}.
Notice that, for 
every cycle $C$ in $G$, 
we have ${\sf mult}(C) = 1$. 
If $C$ is a cycle in $G$, then 
we define 
\begin{equation*}
V(C) \coloneqq \{v_1,v_2,\dots,v_k\}, \ \ 
E(C) \coloneqq \{\{v_{i-1},v_i\} \mid i \in [k]\}.
\end{equation*} 

A set $\mathcal{C}$ of cycles in $G$
is said to be \emph{vertex-disjoint} 
if
we have 
$V(C) \cap V(C^{\prime}) = \emptyset$
for every pair of distinct cycles $C,C^{\prime} \in \mathcal{C}$. 
For each subset $F \subseteq E$ and 
each set $\mathcal{C}$ of vertex-disjoint cycles in $G$, 
$\mathcal{C}$ is called a \emph{decomposition of $F$} if 
$\bigcup_{C \in \mathcal{C}}E(C) = F$. 

For each closed walk 
$(v_0,v_1,\dots,v_k)$ in $G$, 
we define $v_{k+1} \coloneqq v_1$ and 
$v_{k+2} \coloneqq v_2$. 

\begin{definition}
A closed walk $(v_0,v_1,\dots,v_{k})$ in $G$
is said to be \emph{dangerous} if, 
for every integer $i \in [k]$, 
we have $\{v_{i+1},v_i\} \succsim_{v_i} \{v_i,v_{i-1}\}$.
\end{definition}

Define $\mathcal{D}$ as the set of dangerous closed 
walks $C$ 
in $G$ such that ${\sf mult}(C) \le 4|E|$.  
Notice that $|\mathcal{D}|$ is finite. 
Then we define ${\bf P}$ as 
the set of vectors $x \in \mathbb{R}_+^E$ satisfying the following 
conditions.

\begin{equation} \label{eq_1:constraint_super} 
x(E(v)) = 1 
\ \ \ \mbox{($\forall v \in V$)}.
\end{equation}
\begin{equation} \label{eq_2:constraint_super} 
x(e) + \sum_{v \in e}x(E[\succ_v e]) \ge 1 
\ \ \ \mbox{($\forall e \in E$)}.
\end{equation}
\begin{equation} \label{eq_3:constraint_super} 
\displaystyle{
\sum_{i \in [k]}
x(E[e_i \succsim_{v_i}] \setminus \{e_{i+1}\}) \le 
\lfloor k/2 \rfloor} 
\ \ \ \mbox{($\forall (v_0,v_1,\dots,v_k) \in \mathcal{D}$)},
\end{equation}
where  
we define $e_i \coloneqq \{v_{i-1},v_i\}$ 
for each integer $i \in [k+1]$.

The goal of this paper is to prove the following theorem.

\begin{theorem} \label{main_thm_super}
There exists a super-stable matching in $G$ if and only if 
${\bf P} \neq \emptyset$. 
\end{theorem} 

If the preferences do not contain ties, 
then Theorem~\ref{main_thm_super} coincides with 
\cite[Theorem~4]{TeoS98}.
Since we can prove that 
the separation problem 
for ${\bf P}$ can be solved in polynomial time
in a similar way as \cite{TeoS98}
(see Section~\ref{section:separation}), 
Theorem~\ref{main_thm_super} 
gives another proof of
the polynomial-time solvability of 
the problem of checking the existence of a super-stable 
matching in a given instance of the stable roommates problem with ties
in a complete graph. 
Furthermore, the proof of 
Theorem~\ref{main_thm_super}
implies that 
if ${\bf P} \neq \emptyset$, then 
a super-stable matching can be constructed from 
an element in ${\bf P}$ in polynomial time. 

\section{Proof of Theorem~\ref{main_thm_super}}

In this section, we give the proof of Theorem~\ref{main_thm_super}.  

For each matching $\mu$ in $G$, 
we define the vector $x_{\mu} \in \{0,1\}^E$ by 
\begin{equation*}
x_{\mu}(e) \coloneqq 
\begin{cases}
1 & \mbox{if $e \in \mu$} \\
0 & \mbox{if $e \in E \setminus \mu$}. 
\end{cases}
\end{equation*}
That is, 
$x_{\mu}$ is the characteristic vector of $\mu$. 

\begin{lemma} \label{lemma:membership}
If there exists a super-stable matching $\mu$ in $G$, then $x_{\mu} \in {\bf P}$. 
\end{lemma}
\begin{proof}
Assume that there exists a super-stable matching $\mu$ in $G$. 

\eqref{eq_1:constraint_super} 
Since $\mu$ is a matching in $G$, 
$x_{\mu}(E(v)) = |\mu(v)| = 1$ for every vertex $v \in V$.

\eqref{eq_2:constraint_super} 
Assume that there exists an edge $f \in E$ such
that \eqref{eq_2:constraint_super} is not satisfied. 
Then \eqref{eq_1:constraint_super} implies that 
\begin{equation*} 
1 
< 
\sum_{v \in f}x_{\mu}(E(v)) - x_{\mu}(f) - \sum_{v \in f}x_{\mu}(E[\succ_v f])
= 
x_{\mu}(f) + \sum_{v \in f}x_{\mu}(E[f \succsim_v] \setminus \{f\}).
\end{equation*}  
Thus, since $x_{\mu} \in \{0,1\}^E$, we have 
\begin{equation} \label{eq_1:lemma:membership}
x_{\mu}(f) + 
\sum_{v \in f}x_{\mu}(E[f \succsim_v] \setminus \{f\}) \ge 2. 
\end{equation}  

\begin{claim} \label{claim_1:lemma:membership} 
$x_{\mu}(f) = 0$, i.e., $f \notin \mu$. 
\end{claim}
\begin{proof}
Assume that $x_{\mu}(f) = 1$, i.e., 
$f \in \mu$. 
Then \eqref{eq_1:lemma:membership}
implies that 
there exist a vertex $v \in f$ and an edge $g \in E(v) \setminus \{f\}$
such that $g \in \mu$.
However, since $v \in f \cap g$ and $f \neq g$, this contradicts the fact that 
$\mu$ is a matching in $G$. 
\end{proof}  

Notice that 
Claim~\ref{claim_1:lemma:membership}
and \eqref{eq_1:lemma:membership}
imply that 
$f \succsim_v \mu(v)$
for every vertex $v \in f$. 
Thus, 
$f$ weakly blocks $\mu$.
This contradicts the fact that 
$\mu$ is super-stable. 

\eqref{eq_3:constraint_super} 
Let $(v_0,v_1,\dots,v_k)$ be a closed walk in $\mathcal{D}$.
Then for each integer $i \in [k+2]$, 
we define  
$e_{i} \coloneqq \{v_{i-1},v_i\}$. 
We first prove that, for every 
integer $i \in [k]$, 
\begin{equation} \label{eq_2:lemma:membership}
x_{\mu}(E[e_i \succsim_{v_i}] \setminus \{e_{i+1}\}) 
+ x_{\mu}(E[e_{i+1} \succsim_{v_{i+1}}] \setminus \{e_{i+2}\}) \le 1.
\end{equation} 
Assume that there exists an integer $j \in [k]$ such that 
\eqref{eq_2:lemma:membership} is not satisfied. 
Since $x_{\mu} \in \{0,1\}^E$, 
\begin{equation} \label{eq_3:lemma:membership}
x_{\mu}(E[e_j \succsim_{v_j}] \setminus \{e_{j+1}\}) 
+ x_{\mu}(E[e_{j+1} \succsim_{v_{j+1}}] \setminus \{e_{j+2}\}) \ge 2.
\end{equation} 

\begin{claim} \label{claim_3:lemma:membership} 
$x_{\mu}(e_{j+1}) = 0$, i.e., $e_{j+1} \notin \mu$. 
\end{claim}
\begin{proof}
Assume that $x_{\mu}(e_{j+1}) = 1$, i.e., 
$e_{j+1} \in \mu$. 
Since $x_{\mu}(E(v)) = 1$ for every vertex $v \in V$, 
\begin{equation*} 
x_{\mu}(E[e_j \succsim_{v_j}] \setminus \{e_{j+1}\}) \ge 1.
\end{equation*}
Thus, there exists an edge 
$f \in E(v_j) \setminus \{e_{j+1}\}$ such that 
$f \in \mu$. 
Since $v_j \in e_{j+1} \cap f$ and 
$f \neq e_{j+1}$, 
this contradicts the fact that $\mu$ is a matching in $G$.
\end{proof}  

Furthermore, \eqref{eq_3:lemma:membership}
implies that 
$e_j \succsim_{v_j} \mu(v_j)$ and 
$e_{j+1} \succsim_{v_{j+1}} \mu(v_{j+1})$.
Since the definition of a dangerous walk implies that 
$e_{j+1} \succsim_{v_j} e_j$, 
we have 
$e_{j+1} \succsim_{v_j} \mu(v_j)$.
Thus, 
Claim~\ref{claim_3:lemma:membership} implies that 
$e_{j+1}$ weakly blocks $\mu$. 
This is a contradiction. 
This completes the proof of \eqref{eq_2:lemma:membership}.

It follows from \eqref{eq_2:lemma:membership} that  
\begin{equation*}
\begin{split}
k & \ge 
\sum_{i \in [k]}
\big(x_{\mu}(E[e_i \succsim_{v_i}] \setminus \{e_{i+1}\}) 
+ x_{\mu}(E[e_{i+1} \succsim_{v_{i+1}}] \setminus \{e_{i+2}\})\big)\\
& =
2\sum_{i \in [k]}
x_{\mu}(E[e_i \succsim_{v_i}] \setminus \{e_{i+1}\}). 
\end{split} 
\end{equation*}
Since $x_{\mu} \in \{0,1\}^E$, this implies that 
\eqref{eq_3:constraint_super} is satisfied.
This completes the proof.  
\end{proof} 

Lemma~\ref{lemma:membership} implies that 
if there exists 
a super-stable matching in $G$, then 
${\bf P} \neq \emptyset$.
Thus, 
we prove the other direction. 
To this end, we first prove the following lemma. 
The proof of (S1) of Lemma~\ref{lemma:self_dual_super} is  
the same as the proof of 
\cite[Lemma~23]{HuG21arxiv} for bipartite graphs.
Furthermore, 
(S2) of Lemma~\ref{lemma:self_dual_super}
follows from (S1) and 
the 
complementary slackness theorem of linear programming
(see, e.g., \cite[Theorem 3.8]{ConfortiCZ14}).
For completeness, we give the proof
of Lemma~\ref{lemma:self_dual_super}.

\begin{lemma} \label{lemma:self_dual_super} 
For every vector $x \in \mathbb{R}_+^E$ satisfying 
\eqref{eq_1:constraint_super}, \eqref{eq_2:constraint_super}
and 
every edge $e \in E$ such that 
$x(e) > 0$, 
the following statements hold. 
\begin{description}
\item[(S1)]
$x(e) + \sum_{v \in e}x(E[\succ_v e]) = 1$. 
\item[(S2)] 
For every vertex $v \in e$ and 
every edge $f \in E[e \sim_v] \setminus \{e\}$,
we have $x(f) = 0$.  
\end{description}
\end{lemma}
\begin{proof} 
Consider the following linear program.
\begin{equation} \label{eq_1:self_dual_super}
\mbox{Maximize} \ \ \ \ 
 x(E)  \ \ \ \ 
\mbox{subject to} \ \ \ \ 
\mbox{\eqref{eq_1:constraint_super}, \eqref{eq_2:constraint_super}}, \ 
x \in \mathbb{R}_+^E. 
\end{equation}
Then the dual problem of \eqref{eq_1:self_dual_super} is described as follows. 
\begin{equation} \label{eq_2:self_dual_super}
\begin{array}{cl}
\mbox{Minimize} \ \ 
&\alpha(V) 
-
\beta(E) \vspace{2mm}\\ 
\mbox{subject to} \ \ 
& 
\displaystyle{
\sum_{v \in e}
\alpha(v)
-
\beta(e)
- 
\sum_{v \in e}
\beta(E[e \succ_v])
\ge 1} \ \ 
\ \ \ \mbox{($\forall e \in E$)} \vspace{1mm}\\
& 
(\alpha,\beta) \in \mathbb{R}^V \times \mathbb{R}_+^E.
\end{array}
\end{equation}
For each feasible solution $x$ to \eqref{eq_1:self_dual_super}, 
we define $\alpha_x \in \mathbb{R}^V$ by 
$\alpha_x(v) \coloneqq x(E(v))$.

\begin{claim} \label{claim:claim_1:lemma:self_dual_super} 
For every feasible solution $x$ to \eqref{eq_1:self_dual_super}, 
$(\alpha_x,x)$ is a feasible solution to 
\eqref{eq_2:self_dual_super}. 
\end{claim}
\begin{proof} 
Let $x$ be a feasible solution to \eqref{eq_1:self_dual_super}.
Then for every edge $e \in E$, 
since $x \in \mathbb{R}_+^E$, 
\begin{equation*}
\begin{split}
& \sum_{v \in e}
\alpha_x(v)
-
x(e)
- 
\sum_{v \in e}
x(E[e \succ_v])
= 
\sum_{v \in e}
x(E(v))
-
x(e)
- 
\sum_{v \in e}
x(E[e \succ_v])\\
& = 
x(e)
+
\sum_{v \in e}
x(E[\succsim_v e] \setminus \{e\})
\ge 
x(e)
+
\sum_{v \in e}
x(E[\succ_v e]) \ge 1.
\end{split}
\end{equation*}
This completes the proof. 
\end{proof} 

\begin{claim} \label{claim:claim_2:lemma:self_dual_super} 
For every feasible solution $x$ to \eqref{eq_1:self_dual_super}, 
$x$ and $(\alpha_x,x)$ are optimal solutions to 
\eqref{eq_1:self_dual_super} and 
\eqref{eq_2:self_dual_super}, respectively.
\end{claim}
\begin{proof} 
For every feasible solution $x$ to \eqref{eq_1:self_dual_super}, 
since 
\begin{equation*}
\alpha_x(V) 
-
x(E)
= 
\sum_{v \in V}
x(E(v))
-
x(E)
= 
2x(E)
-
x(E)
=
x(E), 
\end{equation*}
Claim~\ref{claim:claim_1:lemma:self_dual_super}
and 
the duality theorem of linear programming
(see, e.g., \cite[Theorem~3.7]{ConfortiCZ14}) imply 
that 
$x$ and $(\alpha_x,x)$ are optimal solutions to 
\eqref{eq_1:self_dual_super} and 
\eqref{eq_2:self_dual_super}, respectively.
\end{proof}

Let $x$ be a vector in $\mathbb{R}_+^E$ satisfying 
\eqref{eq_1:constraint_super} and \eqref{eq_2:constraint_super}.
Then Claim~\ref{claim:claim_2:lemma:self_dual_super}
implies that 
$x$ and $(\alpha_x,x)$ are optimal solutions to 
\eqref{eq_1:self_dual_super} and 
\eqref{eq_2:self_dual_super}, respectively.
Let 
$e$ be an edge in $E$ such that
$x(e) > 0$.  
The 
complementary slackness theorem of linear programming
(see, e.g., \cite[Theorem 3.8]{ConfortiCZ14}) 
implies (S1) and 
\begin{equation*}
\sum_{v \in e}
\alpha_x(v)
-
x(e)
- 
\sum_{v \in e}
x(E[e \succ_v])
=
x(e)
+
\sum_{v \in e}
x(E[\succsim_v e] \setminus \{e\}) = 1.
\end{equation*}
These statements imply that 
\begin{equation*}
\sum_{v \in e}
x(E[\sim_v e] \setminus \{e\}) = 0.
\end{equation*} 
Thus, since $x \in \mathbb{R}_+^E$, 
this implies (S2).
\end{proof} 

In what follows, we fix a feasible solution $x$ to 
${\bf P}$. 
Then we prove that we can construct a super-stable matching in $G$ 
from $x$.
This completes the proof of Theorem~\ref{main_thm_super}.

For each vertex $v \in V$, we define $S(v)$ 
as the set of edges $e \in E(v)$ such that 
$x(e) > 0$. 

\begin{lemma} \label{lemma:strict_super}
For every vertex $v \in V$ and 
every pair of 
distinct edges $e,f\in S(v)$, 
exactly one of $e \succ_v f$, $f \succ_v e$ holds.
\end{lemma}
\begin{proof}
If there exist a vertex $v \in V$ and 
distinct edges $e,f \in S(v)$ 
such that $e \sim_v f$, then 
since $e,f\in S(v)$, this contradicts 
(S2) of Lemma~\ref{lemma:self_dual_super}.
\end{proof}

\begin{lemma} \label{lemme:uniqueness}
For every vertex $v \in V$,
an edge $e \in S(v)$ satisfying 
\begin{equation} \label{eq:def_d_super}
x(E[\succ_v e]) < 1/2 
\le x(E[\succsim_v e]) 
\end{equation} 
is uniquely determined. 
\end{lemma}
\begin{proof}
Let $v$ be a vertex in $V$.
Then \eqref{eq_1:constraint_super}
implies that 
there exists an edge in $S(v)$ satisfying \eqref{eq:def_d_super}. 
Assume that there exist distinct edges $e,f \in E(v)$ 
satisfying \eqref{eq:def_d_super}.
Then Lemma~\ref{lemma:strict_super} implies that 
exactly one of 
$e \succ_v f$, $f \succ_v e$ holds. 
Assume that $e \succ_v f$.
Then 
$x(E[\succ_v f]) \ge x(E[\succsim_v e]) \ge 1/2$.
This contradicts the fact that 
$f$ satisfies \eqref{eq:def_d_super}. 
We can prove the case where 
$e \succ_v f$ in the same way.
This completes the proof.  
\end{proof} 

For each vertex $v \in V$,
we define $d(v)$
as the edge $e \in S(v)$ satisfying 
\eqref{eq:def_d_super}.
Lemma~\ref{lemme:uniqueness}
implies that 
$d(v)$ is well-defined.
For each vertex $v \in V$,
if $d(v) = \{v,w\}$, 
then we define 
$m(v) \coloneqq w$.

\begin{lemma} \label{lemma:edge}
Let $v$ be a vertex in $V$ such that 
$x(E[\succsim_v d(v)]) > 1/2$.
Then 
$m(m(v)) = v$. 
\end{lemma} 
\begin{proof}
Define $w \coloneqq m(v)$.
Then $d(v) = \{v,w\}$. 
Assume that $m(w) \neq v$. 
Since $d(w) \neq d(v)$,
we have at least one of 
$x(E[\succsim_w d(v)]) < 1/2$, 
$x(E[\succ_w d(v)]) \ge 1/2$. 

If $x(E[\succsim_w d(v)]) < 1/2$, then 
since 
\eqref{eq:def_d_super}
implies that 
$x(E[\succ_v d(v)]) < 1/2$, 
\begin{equation*}
1 > 
x(E[\succ_v d(v)]) + 
x(E[\succsim_w d(v)])
\ge x(d(v)) + 
x(E[\succ_v d(v)]) + 
x(E[\succ_w d(v)]).
\end{equation*}
However, 
this contradicts \eqref{eq_2:constraint_super}.
Furthermore, if $x(E[\succ_w d(v)]) \ge 1/2$, then 
since the assumption of this lemma 
implies that 
$x(E[\succsim_v d(v)]) > 1/2$, 
we have 
\begin{equation*}
1 < 
x(E[\succsim_v d(v)]) + 
x(E[\succ_w d(v)])
= x(d(v)) + 
x(E[\succ_v d(v)]) + 
x(E[\succ_w d(v)]),
\end{equation*}
where the equation follows from 
(S2) of Lemma~\ref{lemma:self_dual_super}. 
However, since $d(v) \in S(v)$, 
this contradicts (S1) of 
Lemma~\ref{lemma:self_dual_super}. 
\end{proof}

Define $V^{\ast}$ as the set of vertices $v \in V$ 
such that $m(m(v)) = v$.
Then 
Lemma~\ref{lemma:edge} implies that, for every 
vertex $v \in V \setminus V^{\ast}$, 
$x(E[\succsim_v d(v)]) = 1/2$. 

\begin{lemma} \label{lemma:cycle}
For every vertex $v \in V$, if 
$x(E[\succsim_v d(v)]) = 1/2$, then 
$x(E[\succ_{m(v)} d(v)]) = 1/2$. 
\end{lemma} 
\begin{proof}
Let $v$ be a vertex in $V$ 
such that $x(E[\succsim_v d(v)]) = 1/2$.
Then since $d(v) \in S(v)$, 
it follows from
(S1) and (S2) of Lemma~\ref{lemma:self_dual_super} that 
\begin{equation*}
x(E[\succsim_v d(v)]) + x(E[\succ_{m(v)} d(v)])
= 
x(d(v)) + x(E[\succ_v d(v)]) + x(E[\succ_{m(v)} d(v)]) 
= 1.
\end{equation*}
Thus, since 
$x(E[\succsim_v d(v)]) = 1/2$, 
we have 
$x(E[\succ_{m(v)} d(v)]) = 1/2$.
\end{proof}

\begin{lemma} \label{lemma:not_ast}
For every vertex $v \in V \setminus V^{\ast}$,
$m(v) \in V \setminus V^{\ast}$. 
\end{lemma} 
\begin{proof}
Let $v$ be a vertex in $V \setminus V^{\ast}$. 
Assume that $m(v) \in V^{\ast}$.
Define 
$w \coloneqq m(v)$ and
$u \coloneqq m(w)$.
Then since $w \in V^{\ast}$, $m(u) = w$. 
Define $e \coloneqq \{v,w\}$ and 
$f \coloneqq \{w,u\}$. 

Lemma~\ref{lemma:cycle} 
implies that 
$x(E[\succ_w e]) = 1/2$.
Since $f = d(w)$, 
$x(E[\succsim_w f]) \ge 1/2$.  
Assume that 
$x(E[\succsim_w f]) > 1/2$.  
If $f \succ_w e$, then 
\begin{equation*}
1/2 
< x(E[\succsim_w f]) 
\le x(E[\succ_w e]) = 1/2. 
\end{equation*}
This is a contradiction. Thus, $e \succsim_w f$.
This implies that 
$x(E[\succ_w f]) \ge 
x(E[\succ_w e]) = 1/2$.
This contradicts the fact that 
$f = d(w)$. 
Thus, $x(E[\succsim_w f]) = 1/2$. 

Since 
$x(E[\succsim_w f]) = 1/2$, 
Lemma~\ref{lemma:cycle} implies that 
$x(E[\succ_u f]) = 1/2$.
However, since 
$f = d(u)$, 
\eqref{eq:def_d_super}
implies that 
$x(E[\succ_u f]) < 1/2$.
This is a contradiction. 
\end{proof} 

\begin{lemma} \label{lemma:unique_p}
For every vertex $v \in V \setminus V^{\ast}$,
there exists exactly one vertex $w \in V \setminus V^{\ast}$
such that $m(w) = v$.
\end{lemma} 
\begin{proof}
Lemma~\ref{lemma:not_ast} implies that 
it suffices to prove that, for every vertex $v \in V \setminus V^{\ast}$, 
the number of vertices $w \in V \setminus V^{\ast}$ such that 
$m(w) = v$ is at most one. 
Let $v$ be a vertex in $V \setminus V^{\ast}$. 
Assume that there exist vertices 
$w,u \in V \setminus V^{\ast}$ such that 
$w \neq u$ and 
$m(w) = m(v) = v$. 

Define $e \coloneqq \{w,v\}$ and 
$f \coloneqq \{u,v\}$. 
Lemma~\ref{lemma:cycle} implies that 
$x(E[\succ_v e]) = x(E[\succ_v f]) = 1/2$. 
Since $e,f \in S(v)$, 
Lemma~\ref{lemma:strict_super}
implies that $e \not\sim_v f$. 
Without loss of generality, 
we assume 
that 
$e \succ_v f$.
In this case, 
$e \notin E[\succ_v e]$ 
and 
$e \in E[\succ_v f]$.
Thus, since $E[\succ_v e] \subseteq E[\succ_v f]$ follows from 
$e \succ_v f$, 
$E[\succ_v e]$ is a subset of 
$E[\succ_v f] \setminus \{e\}$. 
Since $x(e) > 0$, 
this implies that
\begin{equation*}
x(E[\succ_v e]) \le 
x(E[\succ_v f] \setminus \{e\}) < 
x(E[\succ_v f]). 
\end{equation*}
However, this contradicts the fact that 
$x(E[\succ_v e]) = x(E[\succ_v f])$. 
\end{proof} 

Let $v$ be a vertex in $V \setminus V^{\ast}$. 
Lemma~\ref{lemma:unique_p} implies that
there exists the unique 
vertex $w \in V \setminus V^{\ast}$
such that $m(w) = v$. 
Define $p(v) \coloneqq w$
and $d^-(v) \coloneqq d(w)$. 

For each vertex $v \in V^{\ast}$, 
we define $d^-(v) \coloneqq d(v)$. 

\begin{lemma} \label{lemma:cycle_2}
For every vertex $v \in V \setminus V^{\ast}$,
we have 
$x(E[\succ_v d^-(v)]) = 1/2$. 
\end{lemma} 
\begin{proof}
Lemma~\ref{lemma:cycle} implies that, 
for every vertex $v \in V \setminus V^{\ast}$,
\begin{equation*}
x(E[\succ_v d^-(v)]) 
= 
x(E[\succ_{m(p(v))} d(p(v))]) 
= 1/2.
\end{equation*}
This completes the proof. 
\end{proof} 

\begin{lemma} \label{lemma:relation_d_p}
For every vertex $v \in V \setminus V^{\ast}$,
we have 
$d(v) \succ_v d^-(v)$. 
\end{lemma} 
\begin{proof}
Let $v$ be a vertex in $V \setminus V^{\ast}$.
Lemma~\ref{lemma:cycle_2}
implies that 
$x(E[\succ_v d^-(v)]) = 1/2$. 
If $d^-(v) \succsim_v d(v)$, then 
since $x(d(v)) > 0$ follows from $d(v) \in S(v)$, 
we have 
\begin{equation*}
x(E[\succsim_v d(v)])
\ge 
x(E[\succ_v d(v)]) + x(d(v)) 
> 
x(E[\succ_v d(v)])
\ge 
x(E[\succ_v d^-(v)])
= 1/2. 
\end{equation*}
However, 
Lemma~\ref{lemma:edge} implies that 
$x(E[\succsim_v d(v)]) = 1/2$.
This is a contradiction.
\end{proof}

Define the subset $L \subseteq E$ as follows. 
For each pair of distinct vertices $v,w \in V$, 
$\{v,w\} \in L$ if and only if 
at least one of $m(v) = w$, $m(w) = v$ holds. 
Then $\bigcup_{e \in L}e = V$. 
By considering the directed graph such that 
its vertex set is $V$ and its arc set contains 
an arc from $v$ to $m(v)$ for each vertex $v \in V$, 
we can prove that 
there exist a set 
$\mathcal{C}$ of vertex-disjoint cycles in $G$ and 
a set $M$ of pairwise disjoint edges in $L$ satisfying the following conditions.
\begin{itemize}
\item
$\bigcup_{e \in M}e = V^{\ast}$.
\item
$\bigcup_{C \in \mathcal{C}}V(C) = V \setminus V^{\ast}$. 
\item
$\mathcal{C}$ is a decomposition of $L \setminus M$. 
\item
For every cycle $(v_0,v_1,\dots,v_k) \in \mathcal{C}$
and every integer $i \in [k]$,
we have $m(v_i) = v_{i+1}$.
\end{itemize}

\begin{lemma} \label{lemma:cycle_L}
Let $C = (v_0,v_1,\dots,v_k)$ be a cycle in $\mathcal{C}$.
\begin{description}
\item[(i)]
For every integer $i \in [k]$, 
$\{v_i,v_{i+1}\} \succ_{v_i} \{v_{i-1},v_i\}$.
\item[(ii)]
$k$ is even. 
\end{description} 
\end{lemma}
\begin{proof}
{\bf (i)}
For every integer $i \in [k]$, 
$\{v_{i-1},v_i\} = d^-(v_i)$ and 
$\{v_i,v_{i+1}\} = d(v_i)$.
This implies that 
(i) follows from 
Lemma~\ref{lemma:relation_d_p}.

{\bf (ii)}
Define 
$e_i \coloneqq \{v_{i-1},v_i\}$ 
for each integer $i \in [k+1]$.
Then for every integer $i \in [k]$, since $e_i = d^-(v_i)$, 
\begin{equation*}
x(E[e_i \succsim_{v_i}] \setminus \{e_{i+1}\}) 
=
x(E[e_i \succsim_{v_i}]) 
=
x(E(v)) - x(E[\succ_{v_i} e_i])
= 1 - 1/2 = 1/2, 
\end{equation*}
where 
the first equation follows from (i), and 
the third equation follows from \eqref{eq_1:constraint_super} and 
Lemma~\ref{lemma:cycle_2}. 
Thus, if $k$ is odd, then 
\begin{equation*}
\sum_{i \in [k]} x(E[e_i \succsim_{v_i}] \setminus \{e_{i+1}\}) 
= k/2 > \lfloor k/2 \rfloor. 
\end{equation*}
However, since (i) implies that 
$C$ is dangerous, 
this contradicts \eqref{eq_3:constraint_super}. 
\end{proof} 

Since 
$\bigcup_{e \in L}e = V$,
Lemma~\ref{lemma:cycle_L}(ii)
implies that we can construct 
a 
matching $\mu$ in $G$ such that 
$\mu \subseteq L$
by taking all the edges in $M$ and 
every other edge along 
each cycle in $\mathcal{C}$.

\begin{lemma} \label{lemma:no_block}
Let $\mu$ be a matching in $G$ 
such that $\mu \subseteq L$, and let 
$e$ be an edge in $L \setminus \mu$.
Then $e$ does not weakly block $\mu$. 
\end{lemma}
\begin{proof}
There exist a cycle 
$(v_0,v_1,\dots,v_k) \in \mathcal{C}$ and 
an integer $i \in [k]$ such that 
$e = \{v_{i-1},v_i\}$. 
In 
this case, $\mu(v_i) = \{v_i,v_{i+1}\}$.
Furthermore,
Lemma~\ref{lemma:cycle_L}(i) implies that 
$\mu(v_i) \succ_{v_i} e$.
This implies that 
$e$ does not weakly block $\mu$. 
\end{proof} 

\begin{lemma} \label{lemma:blocking}
Let $\mu$ be a matching in $G$ 
such that $\mu \subseteq L$, and
let 
$e$ be an edge in $E \setminus L$.
Assume that 
$e$ weakly blocks $\mu$.
Then for every vertex $v \in e$, 
$d(v) \succ_v e \succsim_v d^-(v)$.
\end{lemma}
\begin{proof}
We first prove the following claims.

\begin{claim} \label{claim_1:lemma:blocking}
For every vertex $v \in e$, $x(E[\succ_v \mu(v)]) \le 1/2$.
\end{claim}
\begin{proof}
Let $v$ be a vertex in $e$. 
Since $\mu \subseteq L$, 
$\mu(v) \in \{d(v), d^-(v)\}$.
If $\mu(v) = d(v)$, then 
\eqref{eq:def_d_super} implies that 
$x(E[\succ_v \mu(v)]) < 1/2$.
Assume that 
$\mu(v) \neq d(v)$. 
Then 
$\mu(v) = d^-(v)$ and 
$d(v) \neq d^-(v)$. 
Since 
$d(v) \neq d^-(v)$,
we have 
$v \in V \setminus V^{\ast}$. 
Thus, Lemma~\ref{lemma:cycle_2}
implies that 
$1/2 
= x(E[\succ_v d^-(v)])
= x(E[\succ_v \mu(v)])$. 
\end{proof} 

\begin{claim} \label{claim_2:lemma:blocking}
For every vertex $v \in e$, 
$x(E[\succ_v e]) \le 1/2$.
\end{claim}
\begin{proof}
Let $v$ be a vertex in $e$. 
Then since $e \succsim_v \mu(v)$, 
we have
$E[\succ_v e] \subseteq E[\succ_v \mu(v)]$.
This implies that 
$x(E[\succ_v e]) \le x(E[\succ_v \mu(v)])$.
In addition, Claim~\ref{claim_1:lemma:blocking}
implies that 
$x(E[\succ_v \mu(v)]) \le 1/2$.
These imply that $x(E[\succ_v e]) \le 1/2$.
\end{proof} 

Let $v$ be a vertex in $e$. 
Since $\mu \subseteq L$, 
$\mu(v) \in \{d(v), d^-(v)\}$.
If $v \in V^{\ast}$, then 
$d(v) = d^-(v)$. 
If $v \in V \setminus V^{\ast}$, then 
Lemma~\ref{lemma:relation_d_p} implies that 
$d(v) \succ_v d^-(v)$.
Thus, 
since $e \succsim_v \mu(v)$,
$e \succsim_v d^-(v)$. 

Assume that 
$e \succsim_v d(v)$. 
If 
$e \succ_v d(v)$, then 
since 
$x(E[\succsim_v e]) \le x(E[\succ_v d(v)]) < 1/2$, 
Claim~\ref{claim_2:lemma:blocking} implies that 
\begin{equation*}
x(e) + x(E[\succ_v e]) + x(E[\succ_w e])
\le 
x(E[\succsim_v e]) + x(E[\succ_w e]) < 1.
\end{equation*}
However, this contradicts \eqref{eq_2:constraint_super}. 

Assume that $e \sim_v d(v)$.
Then since $e \in E \setminus L$ and $d(v) \in L$,
we have $e \neq d(v)$.
Thus, since $d(v) \in S(v)$, 
it follows from
(S2) of Lemma~\ref{lemma:self_dual_super}
that 
$x(e) = 0$. 
Furthermore, since 
$e \sim_v d(v)$,
we have 
$E[\succ_v e] = E[\succ_v d(v)]$.
Thus, 
since $x(E[\succ_v d(v)]) < 1/2$, 
Claim~\ref{claim_2:lemma:blocking}
implies that 
\begin{equation*}
x(e) + x(E[\succ_v e]) + x(E[\succ_w e])
= 
x(E[\succ_v d(v)]) + x(E[\succ_w e]) < 1.
\end{equation*}
However, this contradicts \eqref{eq_2:constraint_super}.
\end{proof} 

Define $B$ as the set of edges $e \in E \setminus L$ 
such that, 
for every vertex $v \in e$, 
$d(v) \succ_v e \succsim_v d^-(v)$.
Then 
Lemmas~\ref{lemma:no_block} and \ref{lemma:blocking}
imply that, for every matching $\mu$ in $G$ such that 
$\mu \subseteq L$, 
if there exists an edge $e \in E \setminus \mu$ that weakly 
blocks $\mu$, then $e \in B$.

\begin{lemma} \label{lemma:not_m}
For every edge $e \in B$ and 
every edge $f \in M$, 
$e \cap f = \emptyset$. 
\end{lemma}
\begin{proof}
Assume that there exist edges $e \in B$ and $f \in M$ such that 
$e \cap f \neq \emptyset$. 
Let $v$ be a vertex in $e \cap f$.
Since $f \in M$, $d(v) = d^-(v)$.
However, this contradicts the definition of 
$B$.
\end{proof} 

Let $C = (v_0,v_1,\dots,v_k)$ be a cycle in $G$.
Then for each integer $i \in [k]$ and 
each closed walk 
$C^{\prime} = (w_0,w_1,\dots,w_{\ell})$ in $G$ 
such that $w_1 = v_i$, 
we define the \emph{closed walk $C^+$ in $G$ obtained 
from $C$ by expanding $v_i$ with 
$C^{\prime}$} by
\begin{equation*}
C^+ \coloneqq 
(v_0,v_1,\dots,v_{i-1},w_1,w_2,\dots,w_{\ell},v_i,v_{i+1},\dots,v_k).
\end{equation*}
Furthermore, for each integer $i \in [k]$ and 
each walk $C^{\prime} = (w_1,w_2,\dots,w_{\ell})$ in $G$ such that 
$w_1 = v_i$ and $w_{\ell} = v_{i+1}$, 
we define the \emph{closed walk $C^+$ in $G$ obtained 
from $C$ by replacing $\{v_i,v_{i+1}\}$ with 
$C^{\prime}$} by
\begin{equation*}
C^+ \coloneqq 
(v_0,v_1,\dots,v_{i-1},v_i,w_2,w_3,\dots,w_{\ell-1},v_{i+1},v_{i+2},\dots,v_k).
\end{equation*}

\begin{lemma} \label{lemma:odd_dangerous}
Let $C = (v_0,v_1,\dots,v_k)$ be a cycle in $G$ 
such that $E(C) \subseteq L \cup B$. 
Then $k$ is even. 
\end{lemma}
\begin{proof}
Assume that $k$ is odd, and we derive a contradiction. 
Since $M$ is a set of pairwise disjoint edges in $L$, 
Lemma~\ref{lemma:not_m} implies that 
$V(C) \cap (\bigcup_{e \in M}e) = \emptyset$. 
First, we construct the closed walk $C^{\ast}$ 
in $G$ by using Algorithm~\ref{alg:odd_cycle}. 

\begin{algorithm}[ht]
Set $C^{\ast}_0 \coloneqq C$.\\
\For{$i = 1,2,\dots,k$}
{
  \uIf{$v_{i} = p(v_{i-1})$}
  {
    Let $C_i = (w_{i,0},w_{i,1},\dots,w_{i,\ell_i})$ be the cycle in $\mathcal{C}$ such that 
    $w_{i,1} = v_{i-1}$ and $w_{i,\ell_i} = v_{i}$.
    Define the walk $C_i^{\prime}$ in $G$ as $(w_{i,1},w_{i,2},\dots,w_{i,\ell_i})$.\\
    Define $C^{\ast}_i$ as the closed walk in $G$ obtained from $C^{\ast}_{i-1}$ by 
    replacing $\{v_{i-1},v_i\}$ with $C_i^{\prime}$.
  }
  \uElseIf{$\{v_{i-1},v_i\} \in B$}
  {
    Let $C_i = (w_{i,0},w_{i,1},,\dots,w_{i,\ell_i})$ be the cycle in $\mathcal{C}$ such that 
    $w_{i,1} = v_{i}$.\\
    Define $C^{\ast}_i$ as the closed walk in $G$ obtained from $C^{\ast}_{i-1}$ by 
    expanding $v_i$ with $C_i$.
  }
  \Else
  {
    Define $C^{\ast}_i \coloneqq C^{\ast}_{i-1}$. 
  }
}
Output $C^{\ast}_k$ as $C^{\ast}$, and halt. 
\caption{Algorithm for constructing a dangerous walk}
\label{alg:odd_cycle}
\end{algorithm}

Assume that 
$C^{\ast} = (c_0,c_1,c_2,\dots,c_q)$. 
For each integer $i \in [q+1]$, 
we define $g_i \coloneqq \{c_{i-1},c_i\}$. 
Then the following statements hold.
\begin{itemize}
\item
Steps~6 to 8 of Algorithm~\ref{alg:odd_cycle}
imply that, for every integer $i \in [q]$, 
$\{g_i,g_{i+1}\} \cap L \neq \emptyset$.
\item
For every integer $i \in [q]$, 
if $g_i \in L$ (resp.\ $g_{i+1} \in L$), then 
$g_i = d^-(c_i)$
(resp.\ $g_{i+1} = d(c_i)$). 
\item
For every integer $i \in [q]$, 
$c_i \in V \setminus V^{\ast}$.
\end{itemize}

\begin{claim} \label{claim_1:lemma:odd_dangerous}
$q$ is odd. 
\end{claim}
\begin{proof}
Lemma~\ref{lemma:cycle_L}(ii)
implies that $\ell_i$ is even
for every integer $i \in [k]$.
Thus, since $k$ is odd, this claim holds. 
\end{proof} 

\begin{claim} \label{claim_6:lemma:odd_dangerous}
For every integer $i \in [q]$, 
$g_i \succsim_{c_i} d^-(c_i)$
and 
$g_{i+1} \succsim_{c_i} d^-(c_i)$
\end{claim}
\begin{proof}
Let $i$ be an integer in $[q]$. 
If $g_i \in L$, then 
$g_i = d^-(c_i)$. 
If $g_i \in B$, then
since $c_i \in g_i$, 
the definition of $B$
implies that 
$g_i \succsim_{c_i} d^-(c_i)$. 
If $g_{i+1} \in L$, then 
$g_{i+1} = d(c_i)$. 
Thus, since 
Lemma~\ref{lemma:relation_d_p}
implies that 
$d(c_i) \succ_{c_i} d^-(c_i)$, 
$g_{i+1} \succ_{c_i} d^-(c_i)$.
If $g_{i+1} \in B$, then
since $c_i \in g_{i+1}$, 
the definition of $B$
implies that 
$g_{i+1} \succsim_{c_i} d^-(c_i)$. 
\end{proof} 

\begin{claim} \label{claim_7:lemma_odd_dangerous}
For every integer $i \in [q]$, 
if $g_{i+1} \in B$ and $g_{i+1} \sim_{c_i} d^-(c_i)$, then 
$x(g_{i+1}) = 0$. 
\end{claim}
\begin{proof}
Let $i$ be an integer in $[q]$. 
Since $d^-(c_i) \in L$ and $g_{i+1} \notin L$, 
we have
$d^-(c_i) \neq g_{i+1}$.
Since $d^-(c_i) \in S(c_i)$,
(S2) of 
Lemma~\ref{lemma:self_dual_super}
implies that $x(g_{i+1}) = 0$. 
This completes the proof. 
\end{proof} 

\begin{claim} \label{claim_4:lemma_odd_dangerous}
For every integer $i \in [q]$, 
$E[d^-(c_i) \succsim_{c_i}] \setminus \{g_{i+1}\}
\subseteq 
E[g_i \succsim_{c_i}] \setminus \{g_{i+1}\}$. 
\end{claim}
\begin{proof}
For every integer $i \in [q]$,
Claim~\ref{claim_6:lemma:odd_dangerous}
implies that 
$E[d^-(c_i) \succsim_{c_i}] 
\subseteq 
E[g_i \succsim_{c_i}]$. 
\end{proof} 

\begin{claim} \label{claim_5:lemma_odd_dangerous}
For every integer $i \in [q]$, 
$x(E[d^-(c_i) \succsim_{c_i}])
= 
x(E[d^-(c_i) \succsim_{c_i}] \setminus \{g_{i+1}\})$.
\end{claim}
\begin{proof}
Let $i$ be an integer in $[q]$. 
Then Claim~\ref{claim_6:lemma:odd_dangerous}
implies that 
$g_{i+1} \succsim_{c_i} d^-(c_i)$. 
If $g_{i+1} \succ_{c_i} d^-(c_i)$,
then $g_{i+1} \notin E[d^-(c_i) \succsim_{c_i}]$.
Thus, we assume that 
$g_{i+1} \sim_{c_i} d^-(c_i)$.
In this case, if $g_{i+1} \in L$, then 
Lemma~\ref{lemma:relation_d_p}
implies that 
$g_{i+1} = d(c_i) \succ_{c_i} d^-(c_i)$. 
This is a contradiction.
Thus, 
$g_{i+1} \notin L$, i.e.,
$g_{i+1} \in B$.
In this case, 
Claim~\ref{claim_7:lemma_odd_dangerous}
implies that $x(g_{i+1}) = 0$.
This completes the proof.
\end{proof} 

\begin{claim} \label{claim_2:lemma:odd_dangerous}
For every integer $i \in [q]$, 
$x(E[g_i \succsim_{c_i}] \setminus \{g_{i+1}\}) \ge 1/2$.
\end{claim}
\begin{proof}
Let $i$ be an integer in $[q]$. 
Recall that 
$c_{i} \in V \setminus V^{\ast}$. 
Thus, 
\begin{equation*}
\begin{split}
x(E[g_i \succsim_{c_i}] \setminus \{g_{i+1}\}) 
& \ge 
x(E[d^-(c_i) \succsim_{c_i}] \setminus \{g_{i+1}\})\\ 
&=
x(E[d^-(c_i) \succsim_{c_i}]) \\
&= 
1 - x(E[\succ_{c_i} d^-(c_i)]) 
= 1/2, 
\end{split}
\end{equation*}
where 
the inequality follows from 
Claim~\ref{claim_4:lemma_odd_dangerous}, 
the first equation follows from 
Claim~\ref{claim_5:lemma_odd_dangerous}, 
the second equation follows from 
\eqref{eq_1:constraint_super}, and 
the third equation follows from 
Lemma~\ref{lemma:cycle_2}. 
\end{proof} 

\begin{claim} \label{claim_3:lemma:odd_dangerous}
$C^{\ast} \in \mathcal{D}$.
\end{claim}
\begin{proof}
Let $i$ be an integer in $[q]$. 
Recall that $\{g_i,g_{i+1}\} \cap L \neq \emptyset$. 
If $g_i,g_{i+1} \in L$, then 
$g_i = d^-(c_i)$ and 
$g_{i+1} = d(c_i)$. 
Thus, in this case, 
Lemma~\ref{lemma:relation_d_p} implies that 
$g_{i+1} \succ_{c_i} g_i$. 
If $g_i \in L$ and $g_{i+1} \in B$, then 
since $g_i = d^-(c_i)$, 
the definition of $B$ implies that 
$g_{i+1} \succsim_{c_i} g_i$. 
Assume that $g_i \in B$ and 
$g_{i+1} \in L$. 
Then $g_{i+1} = d(c_i)$.
Thus, 
the definition of $B$ implies that 
$g_{i+1} \succ_{c_i} g_i$. 

Since $C$ is a cycle in $G$, 
$k \le |V|$. 
Thus, since 
$C_i^{\prime}$ is a cycle in $G$ 
for every integer $i \in [k]$,
we have ${\sf mult}(C^{\ast}) \le 2|V| \le 4|E|$. 
This implies that 
$C^{\ast} \in \mathcal{D}$.
This completes the proof. 
\end{proof} 

Claims~\ref{claim_1:lemma:odd_dangerous}
and \ref{claim_2:lemma:odd_dangerous} imply that 
\begin{equation*}
\sum_{i \in [q]}x(E[g_i \succsim_{c_i}] \setminus \{g_{i+1}\}) \ge q/2 > \lfloor q/2 \rfloor. 
\end{equation*}
Claim 
\ref{claim_3:lemma:odd_dangerous} implies that 
this contradicts \eqref{eq_3:constraint_super}. 
This completes the proof. 
\end{proof} 

We now ready to prove Theorem~\ref{main_thm_super}. 
Define the subgraph $H$ of $G$ by 
$H \coloneqq (V,L \cup B)$.  
Then Lemma~\ref{lemma:odd_dangerous}
implies that 
$H$ is bipartite. 
Thus, 
there exists a partition $\{P,Q\}$ of $V$ 
such that, for every edge $e \in L \cup B$, 
$|e \cap P| = |e \cap Q| = 1$. 

\begin{lemma} \label{lemma:bipartite}
Let $(v_0,v_1,\dots,v_k)$ be a cycle in $\mathcal{C}$.
Then one of the following statements holds.
\begin{itemize}
\item 
$\{v_1,v_3,\dots,v_{k-1}\} \subseteq P$ 
and 
$\{v_2,v_4,\dots,v_k\} \subseteq Q$.
\item 
$\{v_1,v_3,\dots,v_{k-1}\} \subseteq Q$
and 
$\{v_2,v_4,\dots,v_k\} \subseteq P$.
\end{itemize}
\end{lemma}
\begin{proof}
This lemma follows from the fact that 
$\{v_{i-1},v_i\} \in L$ for every integer $i \in [k]$. 
\end{proof} 

Define  
$\mu \coloneqq \{d(v) \mid v \in P\}$.
Then Lemma~\ref{lemma:bipartite}
implies that 
$\mu$ is a matching in $G$.
Thus, what remains is to prove that $\mu$ is super-stable. 
For every edge $e \in B$, 
since $|e \cap P|=1$, 
there exists a vertex $v \in e$ such that 
$\mu(v) = d(v)$. 
Thus, 
since the definition of $B$ implies that 
$d(v) \succ_v e$, 
$\mu$ is super-stable. 
This completes the proof of Theorem~\ref{main_thm_super}.

\section{Polynomial-Time Solvability} 
\label{section:separation}

Here we prove that the separation problem for ${\bf P}$
can be solved in polynomial time. 
The proof follows the proof in 
\cite{TeoS98} for preferences without ties.
The remarkable difference is that we need 
(S2) of
Lemma~\ref{lemma:self_dual_super} to prove the non-negativity 
of the cost function on an auxiliary directed graph
(see Lemma~\ref{lemma:non_negative_cost}).  

Let $x$ be a vector in $\mathbb{Q}_+^E$. 
It is not difficult to see that 
we can determine whether $x$ satisfies \eqref{eq_1:constraint_super} and 
\eqref{eq_2:constraint_super} in polynomial time.
Assume that 
$x$ satisfies \eqref{eq_1:constraint_super} and 
\eqref{eq_2:constraint_super}. 
In what follows, under this assumption, 
we consider the problem 
of determining whether $x$ satisfies \eqref{eq_3:constraint_super}.

Define the auxiliary directed graph 
$D = (N,K)$ as follows. 
Define the vertex set $N$ of $D$ by
\begin{equation*}
N \coloneqq 
\{\langle v,w,s \rangle, \langle w,v,s \rangle \mid \{v,w\} \in E, 
s \in V \setminus \{v,w\}\}
\cup 
\{\langle v,w,v \rangle, \langle w,v,w \rangle \mid \{v,w\} \in E\}. 
\end{equation*}
Then for each pair of distinct vertices $\langle v,w,s \rangle, \langle p,q,r \rangle \in N$, 
there exists an arc 
$(\langle v,w,s \rangle, \langle p,q,r \rangle)$ from
$\langle v,w,s \rangle$ to 
$\langle p,q,r \rangle$ in $K$ if and only if 
$w = p$, $s = q$, and $\{p,q\} \succsim_w \{v,w\}$. 
For each arc $(\langle v,w,s \rangle, \langle w,s,r \rangle) \in K$, 
we define the cost ${\sf cost}(\langle v,w,s \rangle, \langle w,s,r \rangle)$ 
by
\begin{equation*}
{\sf cost}(\langle v,w,s \rangle, \langle w,s,r \rangle) \coloneqq 
1 - x(E[\{v,w\} \succsim_w] \setminus \{w,s\}) - x(E[\{w,s\} \succsim_s] \setminus \{s,r\}).
\end{equation*}

\begin{lemma} \label{lemma:non_negative_cost}
Let $(\langle v,w,s \rangle, \langle w,s,r \rangle)$ be an arc in $K$.
Then 
${\sf cost}(\langle v,w,s \rangle, \langle w,s,r \rangle) \ge 0$. 
\end{lemma}
\begin{proof}
Define $e \coloneqq \{v,w\}$, 
$f \coloneqq \{w,s\}$, and $g \coloneqq \{s,r\}$.
Recall that 
$x$ satisfies \eqref{eq_1:constraint_super} and 
\eqref{eq_2:constraint_super}.
Thus, 
\eqref{eq_2:constraint_super} for $f$
implies that 
\begin{equation} \label{eq_1:lemma:non_negative_cost}
\begin{split}
0 
& \le x(f) + x(E[\succ_w f])
+
x(E[\succ_s f]) - 1\\
& =
x(f) + (1 - x(E[f \succsim_w]))
+
(1 - x(E[f \succsim_s])) - 1\\
& =
1 - x(E[f \succsim_w])
- x(E[f \succsim_s]) 
+ x(f),
\end{split}
\end{equation}
where
the first equation follows from 
\eqref{eq_1:constraint_super}.
If $x(f) = 0$, then 
\eqref{eq_1:lemma:non_negative_cost} 
implies that 
\begin{equation*} 
\begin{split}
0 
& \le 
1 - x(E[f \succsim_w])
- x(E[f \succsim_s]) \\
& \le
1 - x(E[e \succsim_w])
- x(E[f \succsim_s]) \\
& \le
1 - x(E[e \succsim_w] \setminus \{f\})
- x(E[f \succsim_s] \setminus \{g\}) 
= {\sf cost}(\langle v,w,s \rangle, \langle w,s,r \rangle),  
\end{split} 
\end{equation*}
where
the second 
inequality follows from $f \succsim_w e$.  
Thus, we assume that $x(f) > 0$. 
In this case, 
(S2) of 
Lemma~\ref{lemma:self_dual_super}
implies that 
$x(E[f \sim_w] \setminus \{f\}) = x(E[f \sim_s] \setminus \{f\}) = 0$. 
Thus, 
\eqref{eq_1:lemma:non_negative_cost} implies that 
\begin{equation} \label{eq_2:lemma:non_negative_cost}
\begin{split}
0 
& \le
1 - x(E[f \succsim_w])
- x(E[f \succsim_s]) 
+ x(f)\\
& =
1 - x(E[f \succ_w]) - x(f)
- x(E[f \succ_s]) - x(f) 
+ x(f)\\
& =
1 - x(E[f \succ_w]) 
- x(E[f \succ_s]) 
- x(f).
\end{split}
\end{equation}
If $f \succ_w e$, then 
\begin{equation*}
\begin{split}
{\sf cost}(\langle v,w,s \rangle, \langle w,s,r \rangle)
& =
1 - x(E[e \succsim_w] \setminus \{f\})
- x(E[f \succsim_s] \setminus \{g\})\\ 
& \ge 
1 - x(E[e \succsim_w]) - x(E[f \succsim_s])\\
& \ge 
1 - x(E[f \succ_w]) - x(E[f \succ_s]) - x(f).
\end{split}
\end{equation*}
If $f \sim_w e$, then 
(S2) of
Lemma~\ref{lemma:self_dual_super}
implies that 
$x(E[e \sim_w] \setminus \{f\}) = 0$.
Thus, since $f \sim_w e$,
\begin{equation*}
\begin{split}
{\sf cost}(\langle v,w,s \rangle, \langle w,s,r \rangle)
& =
1 - x(E[e \succsim_w] \setminus \{f\})
- x(E[f \succsim_s] \setminus \{g\})\\ 
& = 
1 - x(E[e \succ_w]) - x(E[f \succsim_s] \setminus \{g\})\\
& \ge 
1 - x(E[e \succ_w]) - x(E[f \succsim_s])\\
& = 
1 - x(E[e \succ_w]) - x(E[f \succ_s]) - x(f)\\
& = 
1 - x(E[f \succ_w]) - x(E[f \succ_s]) - x(f).
\end{split}
\end{equation*}
Thus, in both cases, 
\eqref{eq_2:lemma:non_negative_cost}
implies that 
${\sf cost}(\langle v,w,s \rangle, \langle w,s,r \rangle) \ge 0$.
This completes the proof. 
\end{proof} 

For a positive integer $k$, 
a sequence $(a_0,a_1,\dots,a_k)$ 
of vertices in $N$ is called a \emph{directed walk in $D$} if 
we have 
$a_{i-1} \neq a_i$ and $(a_{i-1},a_i) \in K$ 
for every integer $i \in [k]$. 
Let 
$W = (a_0,a_1,\dots,a_k)$ be a directed walk in $D$. 
If $a_0 = a_k$, then $W$ is called 
a \emph{closed directed walk in $D$}. 
Furthermore, if $W$ is a closed directed walk in $D$ and 
$a_i \neq a_j$ holds 
for every pair of distinct integers $i,j \in [k]$, then 
$W$ is called a \emph{directed cycle in $D$}.
If $k$ is odd, then 
$W$ is called 
an \emph{odd closed directed walk in $D$}
or 
an \emph{odd directed cycle in $D$}. 
If $W$ is a closed directed walk in $D$,
then  
we define the cost ${\sf cost}(W)$ of $W$ by 
${\sf cost}(W) \coloneqq \sum_{i=1}^k{\sf cost}(a_{i-1},a_i)$.

\begin{lemma} \label{lemma:separation_equivalence}
$x$ satisfies \eqref{eq_3:constraint_super} 
if and only if 
${\sf cost}(W) \ge 1$
for every odd directed cycle $W$ in $D$. 
\end{lemma}
\begin{proof}
Assume that there exists 
an odd directed cycle $W$ in $D$ such that 
${\sf cost}(W) < 1$, 
and we prove that 
$x$ does not satisfy \eqref{eq_3:constraint_super}. 
Assume that $W = (a_0,a_1,\dots,a_k)$ and 
$a_i = \langle v_{i-1},v_i,v_{i+1} \rangle$
for each integer $i \in [k]$.
Notice that, 
since $(a_{i-1},a_i) \in K$
for every integer $i \in [k]$, 
$v_1,v_2,\dots,v_k$ are well-defined and 
we have $v_0 = v_k, v_{k+1} = v_1$. 
Define  
$C \coloneqq (v_0,v_1,\dots,v_k)$
and $e_i \coloneqq \{v_{i-1},v_i\}$ 
for each integer $i \in [k+2]$. 

\begin{claim} \label{claim_1:lemma:separation_equivalence}
${\sf mult}(C) \le 2|V| \le 4|E|$.
\end{claim}
\begin{proof}
Since 
the number of vertices $\langle v,w,s \rangle \in N$ such that 
$\{v,w\} = e$ is at most $2|V|$
for every edge $e \in E$, 
this claim follows from the fact that 
$W$ is a directed cycle in $D$. 
\end{proof}

Since 
$e_{i} \succsim_{v_i} e_{i-1}$ 
follows from 
$(a_{i-1},a_{i}) \in K$ 
for every 
integer $i \in [k]$, 
Claim~\ref{claim_1:lemma:separation_equivalence}
implies that 
$C \in \mathcal{D}$. 
Furthermore, 
\begin{equation*} 
\begin{split}
1 > {\sf cost}(W) 
& = \sum_{i \in [k]} \big(1 - x(E[e_{i} \succsim_{v_{i}}] \setminus \{e_{i+1}\}) 
- x(E[e_{i+1} \succsim_{v_{i+1}}] \setminus \{e_{i+2}\})\big)\\
& = k - 2 \sum_{i \in [k]} x(E[e_i \succsim_{v_i}] \setminus \{e_{i+1}\}). 
\end{split} 
\end{equation*}
Thus, since 
$k$ is odd, 
\begin{equation*} 
\sum_{i \in [k]} x(E[e_i \succsim_{v_i}] \setminus \{e_{i+1}\})
> (k-1)/2 = \lfloor k/2 \rfloor. 
\end{equation*}
This implies that 
$x$ does not satisfy \eqref{eq_3:constraint_super}. 

Assume that 
$x$ does not satisfy \eqref{eq_3:constraint_super}, i.e., 
there exists a closed walk $C = (v_0,v_1,\dots,v_k)\in \mathcal{D}$ such that 
\begin{equation} \label{eq_1:lemma:separation_equivalence}
\sum_{i \in [k]} x(E[e_i \succsim_{v_i}] \setminus \{e_{i+1}\})
> \lfloor k/2 \rfloor, 
\end{equation}
where we define 
$e_i \coloneqq \{v_{i-1},v_i\}$
for each integer $i \in [k+1]$. 
For each integer $i \in [k]$, 
we define $a_i \coloneqq \langle v_{i-1},v_i,v_{i+1}\rangle$. 
Define $a_0 \coloneqq a_k$
and $W \coloneqq (a_0,a_1,\dots,a_k)$.
Since 
$\{v_i,v_{i+1}\} \succsim_{v_i} \{v_{i-1},v_i\}$ 
for every integer $i \in [k]$, 
$W$ is a closed directed walk in $D$. 
Define $e_{k+2} \coloneqq \{v_1,v_2\}$.
Lemma~\ref{lemma:non_negative_cost}
implies that 
\begin{equation} \label{eq_2:lemma:separation_equivalence}
\begin{split}
0 & \le 
\sum_{i \in [k]}
{\sf cost}(\langle v_{i-1},v_i,v_{i+1}\rangle, \langle v_{i},v_{i+1},v_{i+2}\rangle)\\
& =
\sum_{i \in [k]}
\big(1 - x(E[e_i \succsim_{v_i}] \setminus \{e_{i+1}\})
- x(E[e_{i+1} \succsim_{v_{i+1}}] \setminus \{e_{i+2}\})\big)\\
& =
k 
- 2\sum_{i \in [k]}
x(E[e_i \succsim_{v_i}] \setminus \{e_{i+1}\}),
\end{split}
\end{equation}
If $k$ is even, then $k/2 = \lfloor k/2 \rfloor$. 
Thus, 
\eqref{eq_1:lemma:separation_equivalence} 
and \eqref{eq_2:lemma:separation_equivalence} imply that 
$k$ is odd. 
This implies that
$W$ is an odd closed directed walk in $D$. 
Furthermore, 
\begin{equation*} 
\begin{split}
{\sf cost}(W) 
&= k - 2 \sum_{i \in [k]} x(E[e_i \succsim_{v_i}] \setminus \{e_{i+1}\})\\
&< k - 2\lfloor k/2 \rfloor
= k - 2((k-1)/2) = 1.
\end{split} 
\end{equation*}
Thus, by using $W$ and Lemma~\ref{lemma:non_negative_cost}, 
we can prove that 
there exists 
an odd directed cycle $W^{\prime}$ in $D$ such that 
${\sf cost}(W^{\prime}) < 1$
as follows. 
Assume that $W$ is not a directed cycle in $D$. 
Let $i$ be the minimum integer 
in $[k]$ such that there exists an integer $j \in [i-1]$ 
satisfying the condition that 
$a_i = a_j$. 
The definition of a directed walk implies that 
$j \neq i-1$.
Then $W^{\circ} = (a_j,a_{j+1},\dots,a_i)$ is a directed cycle in 
$D$. 
Since ${\sf cost}(W) < 1$,
Lemma~\ref{lemma:non_negative_cost}
implies that ${\sf cost}(W^{\circ}) < 1$. 
If $i - j$ is odd, then the proof is done. 
Assume that $i - j$ is even. 
Define the closed directed walk $W^{\bullet}$ in $D$ by 
\begin{equation*}
W^{\bullet} \coloneqq (a_0,a_1,\dots,a_{j-1},a_i,a_{i+1},\dots,a_k). 
\end{equation*}
Then Lemma~\ref{lemma:non_negative_cost}
implies that ${\sf cost}(W^{\bullet}) < 1$.
Since $i - j$ is even and 
$W^{\bullet}$ is obtained from $W$ by removing $i - j$ vertices, 
$W^{\bullet}$ consists of an odd number of vertices.
Thus, by repeating this, we can obtain a desired odd directed cycle in $D$.
This completes the proof. 
\end{proof} 

Lemma~\ref{lemma:separation_equivalence}
implies that 
if we can find 
an odd directed cycle $W$ in $D$ that minimizes  
${\sf cost}(W)$ in polynomial time, then 
the separation problem for ${\bf P}$ 
can be solved in polynomial time. 
Since it is known that 
Lemma~\ref{lemma:non_negative_cost}
implies that
this problem can be solved in polynomial time
(see \cite[Problem~8.3.6]{GrotschelLS93}), 
the separation problem for ${\bf P}$ 
can be solved in polynomial time. 

Finally, we consider 
the problem of finding a 
super-stable matching in $G$ if one exists. 
First, we determine whether ${\bf P} \neq \emptyset$. 
If ${\bf P} = \emptyset$, then 
there does not exist a super-stable matching in $G$. 
Assume that ${\bf P} \neq \emptyset$. 
Then we compute an element $x \in {\bf P}$ in polynomial time~\cite{GrotschelLS93}. 
For each vertex $v \in V$, 
we can compute $d(v)$ in polynomial time by using $x$. 
Furthermore, 
we can compute $L,B$ in polynomial time 
by using $d(\cdot)$. 
We can compute a partition $P,Q$ by 
the breadth-first search in polynomial time. 
These imply that we can find a super-stable matching in $G$
in polynomial time. 

\bibliographystyle{plain}
\bibliography{super_roommates_lp_bib}

\end{document}